\documentclass[a4paper,twoside]{amsart}
\usepackage{amsmath}
\usepackage{amsfonts}
\usepackage{amssymb}
\usepackage{amsthm}
\usepackage{newlfont}
\usepackage{graphicx}
\usepackage{amscd}

\theoremstyle{plain}
\newtheorem{Th}{Theorem}[section]
\newtheorem{Cor}[Th]{Corollary}
\newtheorem{Lem}[Th]{Lemma}
\newtheorem{Prop}[Th]{Proposition}
\theoremstyle{definition}

\theoremstyle{remark}
\newtheorem*{Rem}{Remark}
\numberwithin{equation}{section}
\newcommand{\PP}{{\mathbb P}}

\newcommand{\DD}{{\mathbb D}}

\newcommand{\ZZ}{{\mathbb Z}}
\newcommand{\VV}{{\mathbb V}}
\newcommand{\RR}{{\mathbb R}}

\begin{document}

\title[The affine Weyl group symmetry of Desargues maps]
{The affine Weyl group symmetry of Desargues maps \\
and of the non-commutative
Hirota-Miwa system}

\author{Adam Doliwa}

\address{Faculty of Mathematics and Computer Science, University of Warmia and
Mazury,
ul.~\.{Z}o{\l}\-nierska~14, 10-561 Olsztyn, Poland}

\email{doliwa@matman.uwm.edu.pl}

\date{}

\begin{abstract}
We study recently introduced Desargues maps, which provide simple
geometric interpretation of the non-commutative
Hirota--Miwa system. We characterize
them as maps of the
$A$-type root lattice into a projective space such that images of vertices of any
basic regular $N$-simplex are collinear. Such a characterization is 
manifestly invariant with
respect to the corresponding affine Weyl group action, which leads to 
related
symmetries of the Hirota--Miwa system. 
\end{abstract}
\maketitle

\section{Introduction}
The Desargues maps, as defined in \cite{Dol-Des}, are maps 
$\phi:\ZZ^N\to\PP^M(\DD)$ of multidimensional integer lattice into 
projective space of dimension $M\geq 2$ over a division ring $\DD$, 
such that for any pair of indices 
$i\ne j$ the points $\phi(n)$, $\phi(n+\boldsymbol{\varepsilon}_i)$ and 
$\phi(n+\boldsymbol{\varepsilon}_j)$ are collinear; here 
$\boldsymbol{\varepsilon}_i = (0, \dots, \stackrel{i}{1}, \dots ,0)$
is the $i$-th element of the canonical basis of $\RR^N$. 
Under mild genericity conditions, by an appropriate choice of homogeneous
coordinates, such maps are described 
\cite{Dol-Des} in terms of solutions 
$\boldsymbol{\phi}:\ZZ^N\to\DD^{M+1}_*$ of the 
linear system (we consider the right vector spaces over $\DD$)
\begin{equation} \label{eq:lin-dKP}
\boldsymbol{\phi}(n+\boldsymbol{\varepsilon}_i) - 
\boldsymbol{\phi}(n+\boldsymbol{\varepsilon}_j) =  
\boldsymbol{\phi}(n) U_{ij}(n),  
\qquad i \ne j \leq N,
\end{equation}
with the corresponding functions $U_{ij}:\ZZ^N\to\DD_*$. 

The linear system \eqref{eq:lin-dKP} is well known in soliton
theory~\cite{DJM-II,Nimmo-NCKP}. Its compatibility condition
is the following nonlinear system 
\begin{align} \label{eq:alg-comp-U}
& U_{ij}(n) + U_{ji}(n) = 0,  \qquad  U_{ij}(n) + U_{jk}(n) + U_{ki}(n) = 0,\\
& \label{eq:U-rho} 
U_{kj}(n)U_{ki}(n+\boldsymbol{\varepsilon}_j) = 
U_{ki}(n) U_{kj}(n+\boldsymbol{\varepsilon}_i),
\end{align}
for distinct triples $i,j,k$, called the non-commutative Hirota--Miwa system
\cite{FWN,FWN-Capel,Nimmo-NCKP}. Equation \eqref{eq:U-rho} allows to introduce 
the potentials $\rho_i:\ZZ^N\to\DD_*$ such that
\begin{equation} \label{eq:def-rho}
U_{ij}(n) = \left[ \rho_i(n)\right]^{-1}
 \rho_i(n+\boldsymbol{\varepsilon}_j).
\end{equation}
When $\DD$ is commutative, i.e. a field, the functions $\rho_i$ can 
be parametrized in terms of a 
single potential $\tau$ (the tau-function)
\begin{equation} \label{eq:r-tau}
\rho_i(n) = (-1)^{\sum_{k>i}n_k}
\frac{\tau(n+\boldsymbol{\varepsilon}_i)}{\tau(n)}.
\end{equation}
Then equations \eqref{eq:alg-comp-U} can be rewritten as the Hirota--Miwa
\cite{Hirota,Miwa} system (called also
the discrete Kadomtsev--Petviashvili (KP) system) 
\begin{equation} \label{eq:H-M}
\tau(n+\boldsymbol{\varepsilon}_i)
\tau(n+\boldsymbol{\varepsilon}_j+\boldsymbol{\varepsilon}_k) - 
\tau(n+\boldsymbol{\varepsilon}_j)
\tau(n+\boldsymbol{\varepsilon}_i+\boldsymbol{\varepsilon}_k) +
\tau(n+\boldsymbol{\varepsilon}_k)
\tau(n+\boldsymbol{\varepsilon}_i+\boldsymbol{\varepsilon}_j) =0,
\qquad i< j < k,
\end{equation}
whose fundamental role in soliton theory
is described, for example, in~\cite{Miwa,KvL,BoKo-N-KP}. 

As it was shown in \cite{Dol-Des} the four dimensional compatibility of the
Desargues maps is equivalent to the celebrated Desargues theorem
\cite{BeukenhoutCameron-H} of the incidence geometry. However, the above
definition of the Desargues maps does not exhibit the well known
symmetry of the Desargues configuration~\cite{Levi}; see also discussion in \cite{Dol-Des}.
In this paper we propose more geometric definition of the Desargues maps, where
instead of the $\ZZ^N$ lattice we use the root lattice $Q(A_N)$. This point of
view allows to see, from the very beginning, the corresponding affine Weyl
group $W(A_N)$ symmetry of the Desargues maps and of the non-commutative
Hirota--Miwa system.

At this point we should mention a recent related
work~\cite{Schief-talk} of W.~K.~Schief, who found a symmetric description 
of the Laplace sequence of two dimensional quadrilateral 
lattices~\cite{Sauer2,DCN} in terms of certain maps 
of the three dimensional lattice of face-centered cubic
combinatorics, which is the root lattice $Q(A_3)$ 
\cite{Schief-private}; see also Section~\ref{sec:P(K-1,K)} for more details.  
We remark that in the integrable discrete geometry
the face-centered cubic closest sphere packing lattice has been 
used in \cite{DNS-4-7,Doliwa-NEEDS-2007} in a different context
of a sub-lattice version of Miwa's discrete BKP equation~\cite{Miwa}. The idea
of checking compatibility, on three dimensional lattices generated by
tetrahedra, of a system of linear equations on triangular 
lattices has been announced to me also by M.~Nieszporski~\cite{Nieszporski-priv}
in the context of a six-point linear problem~\cite{Nieszporski-6p}, 
which contains as reductions both the four point and the three point linear 
problems.

We remark that the $A_{K-1}$ root lattice served in \cite{DMMMS,TQL} as the
parametrization space of the Laplace transformations of $K$-dimensional 
quadrilateral lattices \cite{MQL}. In this context it is present explicitly in
\cite{Dol-Des}, where it was shown that the theory of $K$-dimensional
quadrilateral lattices and their Laplace transformations is the same as the
theory of $(2K-1)$-dimensional Desargues maps. In
describing this equivalence it was convenient to introduce implicitly
the $A_{2K-1}$ root lattice, but the full geometric flavor of this change
of variables has not been observed there.

Finally, it is worth to mention that our research has been also motivated by 
(extended) affine Weyl groups symmetries of discrete Painlev\'{e} 
systems~\cite{NY-AN,RGO,Noumi,KNY-qKP}. 
The geometric picture associated to the theory of  Painlev\'{e} equations 
is connected usually with the representation theory 
of the affine Weyl groups in terms of birational actions on rational
surfaces~\cite{Sakai}. It is well known 
\cite{AblowitzSegur,FWN-Pap,FWN-RGO,GrammaticosRamani,KNY-qKP} that the 
Painlev\'{e} type
equations can be obtained as symmetry reductions of soliton systems.
As the present work brings to the light the affine Weyl group symmetry already
on the level of the 
Hirota--Miwa system, it can be considered as a prologue to the incidence 
geometric description of the integrability of the Painlev\'{e} type equations.

Notice that, by multidimensional compatibility of the Desargues maps, the 
dimension $N$ of the root lattice
can be arbitrarily large. Therefore our paper provides 
geometric explanation of the appearance of the $A_\infty$ root lattice in the 
Kadomtsev--Petviashvili hierarchy, 
which is encoded \cite{Miwa} in the Hirota--Miwa system, other then the standard
one \cite{DKJM,JM} via the theory of representations of the Lie algebra 
$\mathfrak{gl}(\infty)$.

The construction of the paper is as follows. We collect first in
Section~\ref{sec:simplex-lattice} some useful facts on the $A_N$~root lattice and
the corresponding affine Weyl group. Then in Section~\ref{sec:H-M-aWg} we 
present the $A_N$ root lattice description of the Desargues maps and their 
corresponding affine Weyl group symmetry. Finally, in Section~\ref{sec:aW-HM}
we study the action of the $A_N$ affine Weyl group on solutions of the
non-commutative Hirota--Miwa system. 

We remark that we do not discuss here nor
the Darboux type transformations of the Desargues maps nor the corresponding
B\"{a}cklund type transformations of the Hirota--Miwa system. 
These results are known on the algebraic level
\cite{Nimmo-KP,GilsonNimmoOhta}, and their
incidence geometry interpretation follows from the known geometric meaning
\cite{TQL} of the corresponding transformations of the quadrilateral lattices.
Because also here, like in many integrable discrete systems \cite{LeBen}, there
is no essential difference between the Darboux transformation and 
a step into new
dimension of the lattice system, the permutability of the transformations
follows from the multidimensional compatibility of the Desargues
maps. The details are of no particular interest in the context of this paper 
and will be presented elsewhere.

\section{The $A_N$ root lattice and its affine Weyl group}
\label{sec:simplex-lattice}
In this Section we recall necessary facts on the 
$A_N$ root lattice, its Delaunay tiles and the
corresponding affine Weyl group action. The subject is fairly standard, see for
example \cite{Bourbaki,Coxeter,ConwaySloane,Humphreys,MoodyPatera}. Notice that in order
to simplify the presentation we adjusted some general formulas of the root
lattices theory to this specific case.
\subsection{The $A_N$ root lattice} \label{sec:AN-lattice}
The $N\geq 2$-dimensional root lattice $Q(A_N)$, where
the terminology comes from theory of
simple Lie algebras \cite{Bourbaki},
is generated by vectors along the
edges of regular $N$-simplex. If we take the vertices of the simplex to be 
the vectors of the canonical (and orthonormal with respect to the standard
scalar product) basis in $\RR^{N+1}$
\begin{equation*}
\boldsymbol{e}_i = (0,\dots , \stackrel{i}{1}, \dots , 0), \qquad 
1\leq i \leq N+1,
\end{equation*}
then the generators are 
\begin{equation} \label{eq:ep-ij}
\boldsymbol{\varepsilon}^i_j = \boldsymbol{e}_i - \boldsymbol{e}_j, \qquad
1\leq i \neq j \leq N+1,
\end{equation} 
which identifies the lattice as the set of all vectors $(m_1, \dots ,
m_{N+1}) \in \ZZ^{N+1}$ of integer coordinates with zero sum
$m_1 + \cdots + m_{N+1} = 0 $. We consider the $Q(A_N)$ 
lattice as embedded in
the $N$-dimensional vector space 
\begin{equation*}
\VV = \{ (x_1, \dots , x_{N+1}) \in \RR^{N+1}| \, x_1 + \cdots + x_{N+1} = 0 \}
\end{equation*}
with scalar product $(\cdot | \cdot)$ inherited from the ambient $\RR^{N+1}$.
The standard basis of the root
lattice, useful from the point of view the action of
the affine Weyl group (see Section~\ref{sec:aff-Weyl}), consists of the so 
called \emph{simple roots}
\begin{equation*}
\boldsymbol{\alpha}_i = \boldsymbol{e}_{i} - \boldsymbol{e}_{i+1}, 
\qquad 1\leq i \leq N.  
\end{equation*}

Define the \emph{fundamental weights} $\boldsymbol{\omega}_1, \dots ,
\boldsymbol{\omega}_N \in \VV$ as the dual to the simple root basis
\begin{equation*}
(\boldsymbol{\omega}_i | \boldsymbol{\alpha}_j ) = \delta_{ij}.
\end{equation*}
Vertices of the weight lattice $P(A_N)$,
\begin{equation*}
P(A_N) = \sum_{i=1}^N \ZZ \boldsymbol{\omega}_i ,
\end{equation*}
are the points of the root lattice $Q(A_N)$ and their translates by the
fundamental weights.

\subsection{The $A_N$ Weyl groups} \label{sec:aff-Weyl}
The \emph{Weyl group} $W_0(A_N)$ is the Coxeter group
generated by reflections $r_i$, $1\leq i \leq N$,
with respect of the
hyperplanes through the origin and orthogonal to the corresponding
simple roots 
\begin{equation} \label{eq:ri-action}
r_i:  \boldsymbol{v} \mapsto \boldsymbol{v} - 
2\frac{(\boldsymbol{v} | \boldsymbol{\alpha}_i )}
{(\boldsymbol{\alpha}_i | \boldsymbol{\alpha}_i )}
\boldsymbol{\alpha}_i.
\end{equation}
The group $W_0(A_N)$ is isomorphic to the symmetric group $S_{N+1}$ which act
permuting the vectors $\boldsymbol{e}_i$, $1\leq i \leq N+1$; the generators
$r_i$ are identified then with the transpositions $\sigma_i=(i,i+1)$. 

Denote by $\tilde{\boldsymbol{\alpha}}$ the highest root 
\begin{equation*} 
\tilde{\boldsymbol{\alpha}} = -\boldsymbol{\alpha}_0 =
\boldsymbol{\alpha}_1 + \cdots + 
\boldsymbol{\alpha}_N = 
\boldsymbol{e}_{1} - \boldsymbol{e}_{N+1}.
\end{equation*} 
The \emph{affine Weyl group} $W(A_N)$ is the Coxeter group
generated by $r_1, r_2, \dots , r_N$ and by an
additional affine reflection $r_0$
\begin{equation}\label{eq:r0-action}
r_0:  \boldsymbol{v} \mapsto \boldsymbol{v} - \left( 1 - 
2\frac{(\boldsymbol{v} | \tilde{\boldsymbol{\alpha}} )}
{(\tilde{\boldsymbol{\alpha}} | \tilde{\boldsymbol{\alpha}} )}\right) 
\tilde{\boldsymbol{\alpha}}.
\end{equation}
In more abstract terms the affine Weyl group $W(A_N)$ is defined by the
generators $r_0$, $r_1$,\dots, $r_N$ and the relations
\begin{equation*}
r_i^2 = 1, \qquad (r_i r_j)^2 = 1 \quad (j\neq i, i\pm 1), \qquad
(r_i r_j)^3 = 1 \quad (j= i\pm 1),
\end{equation*}
where indices are considered modulo $N+1$.

From \eqref{eq:ri-action}-\eqref{eq:r0-action} we can obtain the following
formulas, which we will need in Section~\ref{sec:aff-rho}:
\begin{equation} \label{eq:AN-action-edges}
r_i(n+\boldsymbol{\varepsilon}^j_{k}) = r_i(n) +  
\boldsymbol{\varepsilon}^{\sigma_i(j)}_{\sigma_i(k)},\qquad n\in Q(A_N),
\end{equation}
where $\sigma_0 = (1,N+1)$.
In particular 
\begin{equation} \label{eq:AN-action-edges-simple}
r_i(n+\boldsymbol{\alpha}_{j}) = r_i(n) + 
\boldsymbol{\alpha}_{j} - a_{ji} \boldsymbol{\alpha}_{i}, \qquad 0\leq
i,j\leq N , \qquad n\in Q(A_N),
\end{equation} 
where 
\begin{equation*}
a_{ij} = ( \boldsymbol{\alpha}_{i} | \boldsymbol{\alpha}_{j})
= \begin{cases}
2 & i=j, \\ -1 & j = i \pm 1 \quad \mod N+1, \\ 0 & \text{otherwise},
\end{cases}
\end{equation*} 
is the Cartan matrix of the affine Weyl group $W(A_N)$.

The convex hull of the vertices $\{0, \boldsymbol{\omega}_1 , \dots
\boldsymbol{\omega}_N \}$ is the \emph{fundamental region} of for the action of
$W(A_N)$ on $\VV$. 
The affine Weyl group is the semidirect product 
\begin{equation}
W(A_N) =  Q(A_N)\rtimes W_0(A_N)
\end{equation}
of the Weyl group by the translations along the root lattice with the action
given by \eqref{eq:ri-action}. 

\subsection{The Delaunay polytopes of the root lattice}

The \emph{holes} in the lattice are the points of $\VV$ that are
locally maximally distant from the lattice. The convex hull of the lattice
points closest to a hole is called the \emph{Delaunay
polytope}. The Delaunay polytopes of the root lattice $Q(A_N)$
form a tessellation 
of $\VV$ into $N$ convex polytopes $P(k,N)$, $k=1,\dots,N$, 
called \emph{ambo-simplices} in
\cite{ConwaySloane-misc} or regular \emph{hypersimplices} in \cite{Grunbaum}. 
\begin{figure}
\begin{center}
\includegraphics[width=8cm]{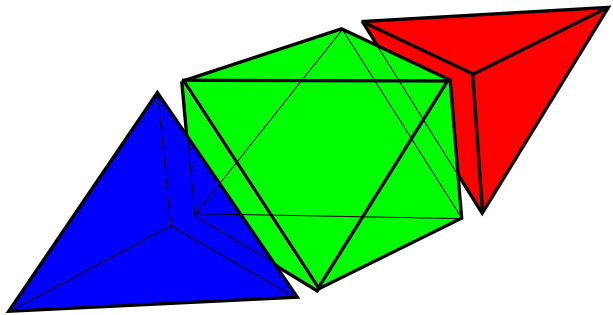}
\end{center}
\caption{The fundamental parallelogram of the $A_3$ root lattice 
decomposed into its Delaunay tiles: two tetrahedra $P(1,3)$ and $P(3,3)$, and
the octahedron $P(2,3)$.}
\label{fig:A3-cube}
\end{figure}
We remark that the tiles $P(1,N)$, congruent to the initial $N$-simplex of
Section~\ref{sec:AN-lattice}, are of particular interest in our paper. They 
will be called basic regular $N$-simplices of the lattice.

The points $Q(A_N)+\boldsymbol{\omega}_k$ are centers of
Delaunay tiles congruent to $P(k,N)$. Up to an appropriate affine 
transformation,
which sends the fundamental parallelogram of the $A_N$ root lattice into the 
standard unit $N$-hypercube $I_N=[0,1]^N$, the tile $P(k,N)$ can be identified
with the region (slice) of $I_N$ between two hyperplanes (see
Figure~\ref{fig:A3-cube})
\begin{equation*}
\hat{P}(k,N) = \{ (x_1, \dots , x_N) \in\RR^N | \; 0\leq x_1, \dots , x_N \leq 1, \;
k-1 \leq x_1 + \cdots + x_N  \leq k \}. 
\end{equation*}

The hypersimplex $P(k,N)$ can be equivalently (up to an affine transformation)
described as the section of the
hypercube $I_{N+1}$ by the corresponding hyperplane 
\begin{equation*}
\tilde{P}(k,N)= \{ (x_1, \dots , x_{N+1}) \in\RR^{N+1} | \; 
0\leq x_1, \dots , x_{N+1} \leq 1, \;
x_1 + \cdots + x_{N+1} = k \}, 
\end{equation*}
where the projection 
$(x_1, \dots , x_N, x_{N+1})\mapsto (x_1, \dots , x_N)$ sends
$\tilde{P}(k,N)$ to $\hat{P}(k,N)$. From the last description it follows that 
the $1$-skeleton of $P(k,N)$ is 
the so called Johnson graph $J(N+1,k)$: its vertices are labelled by 
$k$-point subsets of
$\{1,2, \dots , N+1\}$, and edges are the pairs of such sets with
$(k-1)$-point intersection. 

The following result
\cite{ConwaySloane,MoodyPatera}
is of the fundamental importance for our paper. 
\begin{Lem} \label{lem:Del-perm}
The affine Weyl group acts on the Delaunay tiling by permuting tiles within each
class $P(k,N)$.
\end{Lem}
\begin{Rem}
The ambo-simplices
$P(k,N)$ and $P(N-k+1,N)$ can be identified by applying a point reflection
symmetry,
which is not an element of the affine Weyl group. For our purposes it is 
important to keep the difference.
\end{Rem}
In \cite{MoodyPatera} one can find a detailed description of 
affine Weyl group orbits of facets of various dimension of the 
Delaunay tiles of root lattices. 
\begin{Cor} \label{cor:facets-Weyl}
The following results are either
explicitly stated in \cite{MoodyPatera} or can be without difficulty 
derived using the
method presented there: 
\begin{enumerate}
\item With respect to the $W(A_N)$-action the $K$-dimensional
facets of Delaunay tiles of $Q(A_N)$ are exactly $K$ ambo-simplices 
$P(k,K)$, $1\leq k \leq K$, where the
notation comes from $Q(A_K)$ sublattices. 
\item No two different regular $N$-simplices $P(1,N)$ in $Q(A_N)$
share more then 
a $0$-dimensional facet, i.e. a vertex of the root lattice; 
in any vertex there meet exactly $N+1$ such simplices. 
All the $2$-facets of  $P(1,N)$ are the $P(1,2)$ equilateral
triangles in the number of
$\left( \begin{array}{c} N+1 \\ 3 \end{array} \right)$.
\item The $K$-facet $P(K-1,K)$ has exactly:\\ (i)
$\left( \begin{array}{c} K+1 \\ 2 \end{array} \right)$ vertices;\\
(ii) $(K+1)\left( \begin{array}{c} K \\ 2  \end{array} \right)$ 
$1$-facets (edges);\\
(iii) $\left( \begin{array}{c} K+1 \\ 3  \end{array} \right)$ $2$-facets
$P(1,2)$, of which exactly $K-1$ meet in any vertex of the tile and no two
share an edge, and
$(K+1)\left( \begin{array}{c} K \\ 3  \end{array} \right)$ $2$-facets
$P(2,2)$;\\
(iv) $\left( \begin{array}{c} K+1 \\ 4  \end{array} \right)$ $3$-facets
$P(2,3)$, which are regular octahedra.
\end{enumerate}
\end{Cor}

\section{The root lattice description of the Desargues maps}
\label{sec:H-M-aWg}
In this Section we study geometrically the Desargues maps and their affine Weyl
group symmetry. Its appearance is natural once their definition is restated in
terms of the root lattice. We also discuss the resulting symmetry of the
generalized Desargues configurations which are responsible for multidimensional
compatibility of the maps. We postpone for the next Section the study of the
corresponding properties of the non-commutative Hirota--Miwa system. 

\subsection{The Desargues maps of the root lattice}
Consider the $\ZZ^N$ coordinates in the root lattice $Q(A_N)$ 
by the following identification 
\begin{equation} \label{eq:Z-Q}
\ZZ^N = \sum_{i=1}^N
\ZZ\boldsymbol{\varepsilon}^{N+1}_i= Q(A_N).
\end{equation}
Then the Desargues map condition~\cite{Dol-Des}, recalled at the beginning of the
Introduction, can be 
formulated as collinearity of images of points labelled by $n$,
$n+\boldsymbol{\varepsilon}^{N+1}_i$, $i=1,\dots, N$. 
Because those are the vertices of a basic $N$-simplex
$P(1,N)$ of the root lattice, then one arises to more 
geometric characterization of Desargues maps, which can eventually
be taken as their definition.
\begin{Prop}
Under the above identification \eqref{eq:Z-Q} of $Q(A_N)$ with $\ZZ^N$ the
Desargues maps are the maps $\phi:Q(A_N)\to\PP^M$ such that the vertices of each
basic $N$-simplex $P(1,N)$ are mapped into collinear points. 
\end{Prop}
\begin{Rem}
The image of a fundamental parallelotope of the root lattice under the Desargues
map was called in \cite{Dol-Des} a Desargues cube. 
\end{Rem}
The above proposition/definition from the very beginning exhibits the 
affine Weyl group symmetry of Desargues maps, what can be immediately inferred
from Lemma~\ref{lem:Del-perm}.

\begin{Th} \label{th:AW-sym}
If $\phi:Q(A_N)\to\PP^M$ is a Desargues map then also for any element $w$ 
of the affine Weyl group $W(A_N)$ the map $\phi \circ w$ is a Desargues map,
where we consider the natural action of $W(A_N)$ on the root lattice
$Q(A_N)$. 
\end{Th}
\begin{Cor}
Because the theory of $K$-dimensional
quadrilateral lattices and their Laplace transformations is 
equivalent \cite{Dol-Des}  to the
theory of $(2K-1)$-dimensional Desargues maps then the 
$W(A_{2K-1})$ affine Weyl group symmetry applies, with appropriate
modifications, to quadrilateral lattices as well.
\end{Cor}

\subsection{Generalized
Desargues configurations} \label{sec:P(K-1,K)}

Consider the image of a $K$ dimensional facet of type $P(K-1,K)$ under
generic Desargues map. The analysis will be made on the basis of results
presented in Corollary~\ref{cor:facets-Weyl}.
First notice that for $K=3$ we obtain the Veblen 
configuration~\cite{BeukenhoutCameron-H}, which is responsible for the
three-dimensional compatibility of the Desargues maps~\cite{Dol-Des}: six 
vertices of the octahedron $P(2,3)$ are mapped into six
coplanar points, while its four facets $P(1,2)$ (the other four are $P(2,2)$)
are mapped into four lines.
\begin{Rem}
In \cite{Schief-talk} the maps of the three-dimensional lattice of face
centered cubic combinatorics, i.e. the $Q(A_3)$ root 
lattice~\cite{Schief-private}, into $\RR^3$ with the
property that vertices of the (bipartite) octahedra are mapped into points of
the Veblen configuration (named there the Menelaus configuration, because in the
affine context it is related to the Menelaus theorem)
have been called the Laplace--Darboux lattices. Moreover, it was shown there
that such maps provide a symmetric description of the Laplace sequences of 
two-dimensional quadrilateral lattices~\cite{Sauer2,DCN}. 
\end{Rem}
By point (3) of
Corollary~\ref{cor:facets-Weyl} the Desargues map
images of $\left( \begin{array}{c} K+1 \\ 2 \end{array} \right)$
vertices of $P(K-1,K)$ and of its 
$\left( \begin{array}{c} K+1 \\ 3 \end{array} \right)$
$2$-facets $P(1,2)$ give a configurations of points and lines, respectively,
which satisfies the following conditions:
\begin{enumerate}
\item every line is incident with exactly three points,
\item every point is incident with exactly $K-1$ lines,
\item it contains exactly $\left( \begin{array}{c} K+1 \\ 4 \end{array} \right)$
Veblen configurations.
\end{enumerate} 
In the case of
$K=4$ we obtain the Desargues configuration, which is responsible for the
four-dimensional compatibility of the Desargues maps~\cite{Dol-Des}.
For $K>4$ we obtain generalizations of the Desargues configuration,
called binomial
configurations and studied in \cite{Herrmann,Levi,Prazmowska2}.

Let us discuss the symmetry of the generalized Desargues
configurations. By properties of the Johnson graphs of the ambo-simplices
$P(K-1,K)$ the points of the configuration are labelled by $(K-1)$-point
subsets of $\{ 1,2, \dots , K+1\}$, while the lines of the configuration are
labelled by the $(K-2)$-point subsets. A point is incident with a line if the
line labels are contained in the point labels. The symmetry group of the $K$-th
generalized Desargues configuration is thus the symmetric group of $(K+1)$
elements $S_{K+1}=W_0(A_K)$; we remark that we do not take into account the 
point-line duality in the special case $K=4$ of the original
Desargues configuration \emph{on the plane}. 
\begin{Rem}
There is a simple geometric description  of the origin of the symmetry of
generalized Desargues configurations, attributed by Coxeter \cite{Coxeter-s-d} 
to Cayley~\cite{Cayley} in the basic case $K=4$. Combinatorially, instead of
$(K-1)$-point subsets of $\{ 1,2, \dots , K+1\}$ to label the points of the
configuration one uses the complementary $2$-point subsets (and $3$-point
subsets to label the lines).

Given $K+1$ points in general
position in the projective space $\PP^K$, consider lines joining pairs of 
the points, and planes through the triplets. Intersected by a generic 
hyperplane in $\PP^K$ the lines and planes give $K$-th generalized Desargues 
configuration of 
points and lines on the hyperplane. The symmetry group permutes the $K+1$
points.
\end{Rem}

\section{The affine Weyl group symmetry of the Hirota--Miwa system}
\label{sec:aW-HM}
Below we transfer geometric considerations of the previous Section
into the language of the Hirota--Miwa system and of its symmetries. 
We remark, that similarly one can study the affine Weyl group symmetry of
the non-commutative
discrete modified KP system and the generalized lattice spin system 
\cite{FWN-Capel} (called also the non-commutative Schwarzian discrete KP 
system \cite{KoSchiefSDS-II}), because they are gauge equivalent~\cite{Dol-Des} to the
Hirota--Miwa system.

\subsection{$\ZZ^N$ sectors of a Desargues map}

Fix $\ZZ^N$-coordinates in the $A_N$ root lattice by the identification 
\eqref{eq:Z-Q}, and consider the linear problem of
the Desargues map \eqref{eq:lin-dKP}
in the Hirota--Miwa gauge adjusted to this choice of basis 
\begin{equation} \label{eq:lin-dKP-N+1}
\boldsymbol{\phi}^{N+1}(n+\boldsymbol{\varepsilon}^{N+1}_i) - 
\boldsymbol{\phi}^{N+1}(n+\boldsymbol{\varepsilon}^{N+1}_j) =  
\boldsymbol{\phi}^{N+1}(n) 
U_{ij}^{N+1}(n),  \qquad 1\leq i \ne j \leq N,
\end{equation}
with the corresponding potentials $\rho^{N+1}_i$, $1=1,\dots,N$, such that
\begin{equation*}
U_{ij}^{N+1}(n) = \left[ \rho^{N+1}_i(n)\right]^{-1}
 \rho^{N+1}_i(n+\boldsymbol{\varepsilon}^{N+1}_j).
\end{equation*}
In choosing the above coordinates we used one of the
$N+1$ tiles $P(1,N)$ meeting in
the point $n\in Q(A_N)$. Another choice of a "sector", i.e. 
a basis $\{ \boldsymbol{\varepsilon}^i_j \}$, where the index $i$ is fixed, and
$j\neq i$, along edges of another 
such a tile, should give a similar linear problem.
The following result shows the gauge and the potentials adjusted to
such an equivalent choice. 
\begin{Lem} \label{lem:lin-i}
The functions $\boldsymbol{\phi}^{i}: Q(A_N) \to \DD^{M+1}_*$
given by 
\begin{equation} \label{eq:phi-i}
\boldsymbol{\phi}^{i}(n) = (-1)^{(n|\boldsymbol{\varepsilon}^{N+1}_i)}
\boldsymbol{\phi}^{N+1}(n)
\left[ \rho^{N+1}_i(n)\right]^{-1},
\end{equation}
satisfy the linear system in the $i$-th sector
\begin{equation} \label{eq:lin-dKP-i}
\boldsymbol{\phi}^{i}(n+\boldsymbol{\varepsilon}^{i}_j) - 
\boldsymbol{\phi}^{i}(n+\boldsymbol{\varepsilon}^{i}_k) =  
\boldsymbol{\phi}^{i}(n) 
U_{jk}^{i}(n),  \qquad i, j , k \quad \text{distinct,}
\end{equation}
where
\begin{equation}
U_{jk}^{i}(n) = \left[ \rho^{i}_j(n)\right]^{-1}
 \rho^{i}_j(n+\boldsymbol{\varepsilon}^{i}_k),
\end{equation}
and the potentials $\rho^{i}_j$ are given by
\begin{equation} \label{eq:rho-ij}
\rho^{i}_j(n) = \begin{cases}
\rho^{N+1}_j(n) \left[ \rho^{N+1}_i(n)\right]^{-1} , & \qquad j \neq N+1,\\
\left[ \rho^{N+1}_i(n)\right]^{-1} , & \qquad j = N+1.
\end{cases}
\end{equation}
\end{Lem}
\begin{proof} 
To obtain \eqref{eq:lin-dKP-i} with $j=N+1$ it is enough to apply definitions
\eqref{eq:phi-i} and \eqref{eq:rho-ij} to equation \eqref{eq:lin-dKP-N+1}.
Then in order to check the case $k=N+1$ one has to demonstrate that 
$U^i_{N+1,j}(n) = - U^i_{j,N+1}(n)$, which follows from the basic relation
$U^{N+1}_{ij}(n) = -U^{N+1}_{ji}(n)$ expressed in terms of the potentials 
$\rho^{N+1}_i(n)$ and $\rho^{N+1}_j(n)$. One uses also the following
simple rules to express vectors of the new basis in terms of the initial
vectors
\begin{equation*}
\boldsymbol{\varepsilon}^i_j
= \boldsymbol{\varepsilon}^{N+1}_j - \boldsymbol{\varepsilon}^{N+1}_i,
\qquad \text{where by definition}
\quad \boldsymbol{\varepsilon}^k_k = \boldsymbol{0}.
\end{equation*} 
Finally, to prove equation \eqref{eq:lin-dKP-i} with
$j,k\neq N+1$ one needs to use two equations of the linear problem
\eqref{eq:lin-dKP-N+1} for the pairs $i,j$ and $i,k$, and the equation
$U^{N+1}_{ij}(n) + U^{N+1}_{ki}(n)= - U^{N+1}_{jk}(n)$ of the 
system \eqref{eq:alg-comp-U}.
\end{proof}
\begin{Cor}
Formulas \eqref{eq:phi-i} and \eqref{eq:rho-ij} are self-consistent,
i.e. in the place of the index $N+1$ one can take an arbitrary index, i.e.
\begin{equation*}
\boldsymbol{\phi}^{i}(n) = (-1)^{(n|\boldsymbol{\varepsilon}^{j}_i)}
\boldsymbol{\phi}^{j}(n)
\left[ \rho^{j}_i(n)\right]^{-1}.
\end{equation*}
In particular one can check that
\begin{equation} \label{eq:ccc}
\rho^i_j(n) \rho^k_i(n) = \rho^k_j(n), \qquad \text{where by definition}
\quad \rho^i_i = 1.
\end{equation}
\end{Cor}
From Lemma \ref{lem:lin-i} we get the following conclusions:
\begin{Cor}
The functions $U^\ell_{ij}$ satisfy the system 
\eqref{eq:alg-comp-U}-\eqref{eq:U-rho}
\begin{align} \label{eq:alg-comp-U-l}
& U^\ell_{ij}(n) + U^\ell_{ji}(n) = 0, \qquad  
U^\ell_{ij}(n) + U^\ell_{jk}(n) + U^\ell_{ki}(n) = 0,\\
& \label{eq:U-rho-l} 
U^\ell_{kj}(n)U^\ell_{ki}(n+\boldsymbol{\varepsilon}^\ell_j) = 
U^\ell_{ki}(n) U^\ell_{kj}(n+\boldsymbol{\varepsilon}^\ell_i),
\end{align}
while the potentials $\rho^i_j$ satisfy corresponding equations
\begin{align} \label{eq:rho-1}
& \rho^{i}_\ell(n) \rho^{\ell}_i(n+\boldsymbol{\varepsilon}^{\ell}_j) +
\rho^{j}_\ell(n) \rho^{\ell}_j(n+\boldsymbol{\varepsilon}^{\ell}_i) =0,\\
\label{eq:rho-2}
& \rho^{i}_\ell(n) \rho^{\ell}_i(n+\boldsymbol{\varepsilon}^{\ell}_j) +
\rho^{j}_\ell(n) \rho^{\ell}_j(n+\boldsymbol{\varepsilon}^{\ell}_k) +
\rho^{k}_\ell(n) \rho^{\ell}_k(n+\boldsymbol{\varepsilon}^{\ell}_i) =0.
\end{align}
\end{Cor}
Notice that 
changing the sector can be understood as a symmetry transformation of
both the
linear and nonlinear systems. Other generators of symmetries are 
translations and permutations of indices within a fixed sector.

\subsection{The affine Weyl group action on the edge potentials}
\label{sec:aff-rho}

Let $E(A_N)$ denote the set of \emph{oriented} edges of the 
root lattice $Q(A_N)$, i.e.
elements of $E(A_N)$ are ordered
pairs $[n, n+ \boldsymbol{\varepsilon}^{i}_{j}]$,
where $n\in Q(A_N)$. Define the function $\rho:E(A_N)\to\DD$ by 
\begin{equation}
\rho([n, n+ \boldsymbol{\varepsilon}^{i}_{j}]) = \rho^i_j(n).
\end{equation}
It is convenient to distinguish the simple root functions 
$\rho^i = \rho^i_{i+1}$, $i=1, \dots, N$, which are attached to the simple 
roots directions. 
One can check, using the condition \eqref{eq:ccc}, that for $i< j$ we have
\begin{equation} \label{eq:r-ij}
\rho^i_j = \rho^{j-1} \dots \rho^i, \qquad \rho^j_i = (\rho^i_j)^{-1}.
\end{equation}
Let us define also the function $\rho^0$ as attached to the direction of the
root $\boldsymbol{\alpha}_0$,
which by \eqref{eq:r-ij} gives
\begin{equation}
\rho^0 = (\rho^N \rho^{N-1} \dots \rho^1)^{-1}.
\end{equation}

Define the action of the affine Weyl group on the functions $\rho^i_j$ through
its action on the oriented edges of the root lattice, i.e.
\begin{equation} \label{def:W-action}
(w . \rho ) ([n, n+ \boldsymbol{\varepsilon}^{i}_{j}]) = 
\rho(w^{-1}[n, n+ \boldsymbol{\varepsilon}^{i}_{j}]).
\end{equation} 
\begin{Prop} \label{prop:ri-action-rho}
The action of the generators $r_i$, $i=0,\dots ,N$, of the affine Weyl group on
the functions $\rho^j_k$, is given by 
\begin{equation} \label{eq:ri-rij}
(r_i . \rho^j_k)(n) = \rho^{\sigma_i(j)}_{\sigma_i(k)}(r_i(n)),
\end{equation}
where $\sigma$'s are the transpositions
$\sigma_i = (i, i+1)$, $i=1,\dots ,N$, and $\sigma_0 = (1,N+1)$.
\end{Prop}
\begin{proof}
The conclusion follows from equations \eqref{eq:AN-action-edges},
\eqref{def:W-action}, and from involutivity of the generators $r_i$.
\end{proof}
By equations \eqref{eq:AN-action-edges-simple} and \eqref{eq:r-ij} we have:
\begin{Cor}
The action of the generators $r_i$, $i=0,\dots ,N$, of the affine Weyl group on
the functions $\rho^j$, $j=0,\dots ,N$, is given by 
\begin{equation}
(r_i . \rho^j)(n) = [(\rho^i)^{-a_{ji}^U}\rho^j (\rho^i)^{-a_{ji}^L}](r_i(n)),
\end{equation}
where $a_{ji}^U$ and $a_{ji}^L$ are the "upper" and the "lower" parts
of the Cartan matrix of the affine Weyl group $W(A_N)$
\begin{equation}
a_{ij}^L = \left[ \begin{array}{rrrrr} 1 & 0 &  &  & -1 \\
-1 & 1 & 0 &  &  \\
 & -1 & 1 &  \ddots &     \\
 & & \ddots& \ddots & 0 \\
0 & & & -1 & 1 
\end{array} \right], \qquad
a_{ij}^U = \left[ \begin{array}{rrrrr} 1 & -1 &  &  & 0 \\
0 & 1 & -1 &  &  \\
 & 0 & 1 &  \ddots  &    \\
 & & \ddots &\ddots & -1 \\
-1 &  & & 0 & 1 
\end{array} \right].
\end{equation}
\end{Cor}
The counterpart of Theorem \ref{th:AW-sym} on the level of the non-commutative
Hirota--Miwa system can be then stated as follows.
\begin{Prop}
Transformations of the potentials $\rho^i_j$ given in
equation \eqref{eq:ri-rij} generate 
the affine Weyl group $W(A_N)$ symmetry of solutions of the 
non-commutative Hirota-Miwa system \eqref{eq:rho-1}-\eqref{eq:rho-2}.
\end{Prop}
\begin{proof}
By Lemma~\ref{lem:lin-i} and Proposition~\ref{prop:ri-action-rho} we infer that
the action of the generators of the affine Weyl group on the edge potentials 
is consistent with the action on the wave functions $\phi^j$ given by 
\begin{equation}
(r_i . \phi^j)(n) = \phi^{\sigma_i(j)}(r_i(n)) .
\end{equation}
\end{proof}

\subsection{The $\tau$-functions} Among the edge potentials $\rho^i_j$ only $N$
of them enter nontrivially into the nonlinear system
\eqref{eq:rho-1}-\eqref{eq:rho-2}. They can be chosen by fixing a sector, i.e.
fixing the index $i$. Another choice is given by the simple root
potentials $\rho^i$, $i=1,\dots , N$. We give below another natural set of such basic
potentials (one of them will be redundant). 

Notice that condition \eqref{eq:ccc} allows to introduce functions
$\tau_i:Q(A_N)\to\DD$, $i=1, \dots , N+1$, such that
\begin{equation}
\rho^i_j(n) = \tau_j(n)\left[ \tau_i(n)\right]^{-1}.
\end{equation}
Equations \eqref{eq:rho-1}-\eqref{eq:rho-2} rewritten in terms of the
$\tau$-functions read
\begin{align} \label{eq:tau-comm}
[\tau_i(n)]^{-1} \tau_i(n + \boldsymbol{\varepsilon}^{\ell}_{j})
[\tau_\ell(n + \boldsymbol{\varepsilon}^{\ell}_{j})]^{-1} +
[\tau_j(n)]^{-1} \tau_j(n + \boldsymbol{\varepsilon}^{\ell}_{i})
[\tau_\ell(n + \boldsymbol{\varepsilon}^{\ell}_{i})]^{-1} = & 0,  \\
\label{eq:tau-comm-Hir}
\begin{split}
[\tau_i(n)]^{-1} \tau_i(n + \boldsymbol{\varepsilon}^{\ell}_{j})
[\tau_\ell(n + \boldsymbol{\varepsilon}^{\ell}_{j})]^{-1} + 
[\tau_j(n)]^{-1} \tau_j(n + \boldsymbol{\varepsilon}^{\ell}_{k}) 
[\tau_\ell(n + \boldsymbol{\varepsilon}^{\ell}_{k})]^{-1}  + &\\
[\tau_k(n)]^{-1} \tau_k(n + \boldsymbol{\varepsilon}^{\ell}_{i}) &
[\tau_\ell(n + \boldsymbol{\varepsilon}^{\ell}_{i})]^{-1} = 0.
\end{split}
\end{align}
It is not difficult to verify the following result.
\begin{Prop}
The action of the affine Weyl group on the edge potentials $\rho^i_j$ follows
from the action on the $\tau$-functions given by
\begin{equation}
(r_i . \tau_j)(n) = \tau_{\sigma_i(j)}(r_i(n)).
\end{equation}
\end{Prop}

\begin{Rem}
Notice that by equation \eqref{eq:phi-i} the function
\begin{equation}
(-1)^{(n|\boldsymbol{e}_i)}\boldsymbol{\phi}^{i}(n) 
\tau_i(n), \qquad i=1,\dots ,N+1,
\end{equation}
is independent of the index $i$.
\end{Rem}
\begin{Rem}
For $\DD$ commutative, i.e. a field, one can resolve equations 
\eqref{eq:tau-comm} by expressing $N$ $\tau$-functions in terms of one of them,
for example
\begin{equation}
\tau_i(n) = (-1)^{\sum_{\ell > i}n^{N+1}_\ell} 
\tau_{N+1}(n+\boldsymbol{\varepsilon}^{N+1}_{i}), \qquad
n = \sum_{\ell =1}^N n^{N+1}_\ell \boldsymbol{\varepsilon}^{N+1}_{\ell},
    \qquad i\neq N+1,
\end{equation} 
compare with equation \eqref{eq:r-tau}. Then the remaining equations 
\eqref{eq:tau-comm-Hir} reduce to the standard form \eqref{eq:H-M} of the Hirota-Miwa
system.
\end{Rem}
\section*{Acknowledgments}
I would like to thank  to Krzysztof Pra\.{z}mowski 
for information on binomial configurations and remarks on their combinatorial
and geometric structure. Special thanks
are to Masatoshi Noumi for his lectures on the elliptic Painlev\'{e} VI 
equation, during the programme \emph{Discrete Integrable Systems} in the
Isaac Newton Institute for Mathematical Sciences, and for subsequent discussions
which arouse again my interest in the incidence geometric characterization of the 
Painlev\'{e} systems and motivated the present research. I~would like also to
thank to an anonymous referee of my earlier paper~\cite{Dol-Des} who insisted on
more geometric characterization of Desargues maps then that given there.

\bibliographystyle{amsplain}

\providecommand{\bysame}{\leavevmode\hbox to3em{\hrulefill}\thinspace}

\end{document}